\pdfoutput=1
\documentclass[letterpaper,USenglish,autoref]{lipics-v2019}

\usepackage{booktabs}   

\usepackage{stmaryrd} 
\usepackage{proof}
\usepackage{mathpartir}

\newlength{\rWidth}

\newcommand{\m}[1]{\mathsf{#1}}
\newcommand{\mb}[1]{\mathbf{#1}}

\newenvironment{sill}{\begin{tabbing}}{\end{tabbing}}

\newcommand{\Sg}{\Sigma}

\newcommand{\G}{\Gamma}

\newcommand{\proves}{\vDash}

\newcommand{\lolli}{\multimap}
\newcommand{\tensor}{\otimes}
\newcommand{\with}{\mathbin{\binampersand}}

\newcommand{\one}{\mathbf{1}}

\newcommand{\semi}{\; ; \;}
\newcommand{\ichoiceop}{\oplus}
\newcommand{\echoiceop}{\with}
\newcommand{\ichoice}[1]{\ichoiceop \{ #1 \}}
\newcommand{\echoice}[1]{\echoiceop \{ #1 \}}

\newcommand{\mi}[1]{\mbox{\it #1}}

\newcommand{\tassertop}{{?}}  
\newcommand{\tassumeop}{{!}}  
\newcommand{\tassert}[1]{\tassertop\{#1\}. \,} 
\newcommand{\tassume}[1]{\tassumeop\{#1\}. \,} 

\newcommand{\tforall}[1]{\forall #1. \, }
\newcommand{\texists}[1]{\exists #1. \, }

\newcommand{\Next}{\raisebox{0.3ex}{$\scriptstyle\bigcirc$}}
\renewcommand{\next}[1]{\Next #1}
\newcommand{\tdelay}[2]{
    \IfEqCase{#2}{%
        {1}{\next{#1}}%
    }[{\Next^{#2} (#1)}]%
}%

\newcommand{\indv}[1]{[\overline{#1}]}

\newcommand{\queue}[1]{\m{queue}_{#1}}

\newcommand{\cons}{\mathcal{C}}
\newcommand{\vars}{\mathcal{V}}

\newcommand{\unfldsg}[2]{\m{unfold}_{#1}(#2)}
\newcommand{\unfold}[1]{\unfldsg{\Sg}{#1}}
\newcommand{\rel}{\mathcal{R}}

\definecolor{verylightgray}{rgb}{.97,.97,.97}

\lstdefinelanguage{Solidity}{
	keywords=[1]{anonymous, assembly, assert, balance, break, call, callcode, case, catch, class, constant, continue, constructor, contract, debugger, default, delegatecall, delete, do, else, emit, event, experimental, export, external, false, finally, for, function, gas, if, implements, import, in, indexed, instanceof, interface, internal, is, length, library, log0, log1, log2, log3, log4, memory, modifier, new, payable, pragma, private, protected, public, pure, push, require, return, returns, revert, selfdestruct, send, solidity, storage, struct, suicide, super, switch, then, this, throw, transfer, true, try, typeof, using, value, view, while, with, addmod, ecrecover, keccak256, mulmod, ripemd160, sha256, sha3}, 
	keywordstyle=[1]\color{blue}\bfseries,
	keywords=[2]{address, bool, byte, bytes, bytes1, bytes2, bytes3, bytes4, bytes5, bytes6, bytes7, bytes8, bytes9, bytes10, bytes11, bytes12, bytes13, bytes14, bytes15, bytes16, bytes17, bytes18, bytes19, bytes20, bytes21, bytes22, bytes23, bytes24, bytes25, bytes26, bytes27, bytes28, bytes29, bytes30, bytes31, bytes32, enum, int, int8, int16, int24, int32, int40, int48, int56, int64, int72, int80, int88, int96, int104, int112, int120, int128, int136, int144, int152, int160, int168, int176, int184, int192, int200, int208, int216, int224, int232, int240, int248, int256, mapping, string, uint, uint8, uint16, uint24, uint32, uint40, uint48, uint56, uint64, uint72, uint80, uint88, uint96, uint104, uint112, uint120, uint128, uint136, uint144, uint152, uint160, uint168, uint176, uint184, uint192, uint200, uint208, uint216, uint224, uint232, uint240, uint248, uint256, var, void, ether, finney, szabo, wei, days, hours, minutes, seconds, weeks, years},	
	keywordstyle=[2]\color{teal}\bfseries,
	keywords=[3]{block, blockhash, coinbase, difficulty, gaslimit, number, timestamp, msg, data, gas, sender, sig, value, now, tx, gasprice, origin},	
	keywordstyle=[3]\color{violet}\bfseries,
	identifierstyle=\color{black},
	sensitive=false,
	comment=[l]{//},
	morecomment=[s]{/*}{*/},
	commentstyle=\color{gray}\ttfamily,
	stringstyle=\color{red}\ttfamily,
	morestring=[b]',
	morestring=[b]"
}

\newcommand{\ins}{\iota}
\newcommand{\tcm}{\mathcal{M}}
\newcommand{\inc}[1]{\m{inc}(#1)}
\newcommand{\dec}[1]{\m{dec}(#1)}
\newcommand{\goto}{\m{goto}}
\newcommand{\zeroc}[1]{\m{zero}(#1) ?}
\newcommand{\halt}{\m{halt}}
\newcommand{\clo}[3]{\langle #1 \semi #2 \semi #3\rangle}

\title{Session Types with Arithmetic Refinements} 

\author{Ankush Das}{Carnegie Mellon University, USA}{}{}{}

\author{Frank Pfenning}{Carnegie Mellon University, USA}{}{}{}

\authorrunning{A. Das and F. Pfenning} 

\Copyright{Ankush Das and Frank Pfenning} 



\category{} 

\relatedversion{} 

\supplement{}

\nolinenumbers 

\hideLIPIcs  

\EventEditors{John Q. Open and Joan R. Access}
\EventNoEds{2}
\EventLongTitle{42nd Conference on Very Important Topics (CVIT 2016)}
\EventShortTitle{CVIT 2016}
\EventAcronym{CVIT}
\EventYear{2016}
\EventDate{December 24--27, 2016}
\EventLocation{Little Whinging, United Kingdom}
\EventLogo{}
\SeriesVolume{42}
\ArticleNo{23}

\begin{document}

\maketitle

\begin{abstract}
  Session types statically prescribe bidirectional communication
  protocols for message-passing processes.  However, simple
  session types cannot specify properties beyond the type of exchanged
  messages. In this paper we extend the type system by using index
  refinements from linear arithmetic capturing intrinsic attributes of
  data structures and algorithms.  We show that, despite the
  decidability of Presburger arithmetic, type equality and therefore
  also subtyping and type checking are now undecidable, which stands in
  contrast to analogous dependent refinement type systems from
  functional languages.  We also present a practical, but incomplete
  algorithm for type equality, which we have used in our implementation
  of Rast, a concurrent session-typed language with arithmetic index
  refinements as well as ergometric and temporal types.  Moreover, if
  necessary, the programmer can propose additional type bisimulations
  that are smoothly integrated into the type equality algorithm.
\end{abstract}

\section{Introduction}
\label{sec:intro}

\emph{Session types}~\cite{Honda93concur,Vasconcelos12iandc}
provide a structured way of prescribing communication protocols in
message-passing systems.  This paper focuses on \emph{binary session
  types} governing the interactions along channels with two endpoints.
They arise either directly as part of a program notation~\cite{Honda98esop}, or
as the result of \emph{endpoint projection} of multi-party session
types~\cite{Honda08POPL} and are thus of central importance in the study of
message-passing concurrency.  Moreover, a Curry-Howard correspondence
relates propositions of linear logic to session
types~\cite{Caires10concur,Wadler12icfp,Caires16mscs}, further
evidence for their fundamental nature.

Once recursion is introduced for session types as well as processes,
we are confronted with the question as to what is the correct notion
of type equality since its use in type checking is inescapable.  Gay
and Hole~\cite{Gay05acta} convincingly answer this question and also
provide a practical algorithm for subtyping (which implies an
algorithm for type equality).  First, since the endpoints of channels
need to agree on a type (or possibly two dual types) for
communication, recursive types should be a priori \emph{structural}
rather than \emph{nominal}.  Second, types should be equal if their
observable communication behaviors are indistinguishable.  This means
that two types should be equal if there is a \emph{bisimulation}
between them.  This is particularly elegant since the definition
is independent of any particular programming language in which session
types are embedded, or whether they are checked statically or dynamically.
The algorithm for type equality then constructs a bisimulation.  It
terminates because the number of pairs of types that might be related
by the bisimulation is finite.

Like any type system, basic session types are limited in the kind of
properties they can express, which has led to some generalizations
such as polymorphic~\cite{Wadler12icfp,Caires13esop,Griffith16phd} and
context-free~\cite{Thiemann16icfp} session types, each with its own
questions for type equality.  In this paper we propose a natural
\emph{linear arithmetic refinement} of session types, which allows us to
capture a number of significant properties of message-passing
communication such as size or value of data structures, number of messages
exchanged or delay in those messages.
In conception, this
refinement is closely related to \emph{indexed types} or
\emph{value-dependent types} familiar from functional
languages~\cite{Zenger97,Xi99popl,Rondon08pldi}, where the indices are
arithmetic expressions.

To our surprise, despite an eminently decidable index domain, the type
equality problem becomes undecidable.  We show this via a reduction
from the non-halting problem for two-counter machines~\cite{TwoCounterMachines}.
Analyzing this reduction in detail shows that the problem is already
undecidable for a single type constructor (pick either internal
($\oplus$) or external ($\with$) choice, in addition to arithmetic
refinements).  While our type system is equirecursive to aid in the
simplicity of programming, even retreating to isorecursive types
leaves the problem undecidable.  Finally, one may be tempted to blame
the quantifiers in Presburger arithmetic, but our reduction shows that
even if we restrict ourselves to linear arithmetic with universal
prefix quantification only, type equality remains undecidable.


A retrenchment to a \emph{nominal} interpretation of recursive types
would rule out too many programs and complicate communications, so we
develop a sound but incomplete algorithm.  Our experience with the
Rast implementation~\cite{Das20fscd} to date shows that it is
effective in practice (see \autoref{sec:impl} for further discussion).


Most closely related is the design of LiquidPi~\cite{Griffith13nfm},
but it refines only basic data types such as $\m{int}$ rather than
equirecursively defined session types.  The resulting system has a
decidable subtyping problem and even type inference (under reasonable
assumptions on the constraint domain), but it cannot express many of
our motivating examples.  Along similar lines, refinements of basic
data types together with subtyping have been proposed for
\emph{runtime monitoring} of binary session-typed
communication~\cite{Gommerstadt18esop,Gommerstadt19phd}.
Label-dependent session types~\cite{Thiemann20popl} also support types
indexed by natural numbers using a fixed schema of iteration with a
particular unfolding equality, rather than arbitrary recursion and
bisimulation.  Zhou et al.~\cite{Zhou19Fluid,ZhouThesis} refine base
types with arithmetic expressions in the context of multiparty session
types without recursive types.  In this simpler setting, they obtain a
decidable notion of type equality.  Further related work can be found
in \autoref{sec:related}.


\section{Basic Session Types}
\label{sec:basic}

We review the basic language of binary session types.  We take the
intuitionistic point of view~\cite{Caires10concur,Caires16mscs},
since our experiments and motivating examples have been carried out in
Rast~\cite{Das20fscd}.  Changes for a classical
view~\cite{Wadler12icfp} are minimal and do not affect our results or
algorithms.  We would add a type $\bot$ dual to $\one$ with only minor
changes to the remainder of the development.
\[
  \begin{array}{lclll}
    A,B,C & ::= & \ichoice{\ell : A_\ell}_{\ell \in L} & \mbox{send label $k \in L$} & \mbox{continue at type $A_k$} \\
          & \mid & \echoice{\ell : A_\ell}_{\ell \in L} & \mbox{receive label $k \in L$} & \mbox{continue at type $A_k$} \\
          & \mid & A \tensor B & \mbox{send channel $a : A$} & \mbox{continue at type $B$} \\
          & \mid & A \lolli B & \mbox{receive channel $a : A$} & \mbox{continue at type $B$} \\
          & \mid & \one & \mbox{send $\m{close}$ message} & \mbox{no continuation} \\
          & \mid & V & \mbox{defined type variable}
  \end{array}
\]
We assume that labels $\ell \in L$ (for a finite, nonempty set $L$)
and $\m{close}$ messages can be observed, but the identity of channels
can not.  Instead any \emph{communication} along channels that are
sent and received can be observed in turn.  Based on this notion, we
adopt \emph{type bisimulations} from Gay and Hole~\cite{Gay05acta}.
Rather than an explicit recursive type constructor $\mu$ we postulate
a \emph{signature} $\Sigma$ with definitions of
type variables $V$.
\[
  \begin{array}{llcl}
    \mbox{Signature} & \Sigma & ::= & \cdot \mid \Sigma, V = A
  \end{array}
\]
In a \emph{valid signature} all definitions $V = A$ are
\emph{contractive}, that is, $A$ is not itself a type variable.  This
allows us to take an \emph{equirecursive} view of type definitions,
which means that unfolding a type definition does not require
communication.  We can easily adapt our definitions to an
\emph{isorecursive} view~\cite{Lindley16icfp,Derakhshan19arxiv} with
explicit $\m{unfold}$ messages (see the remark at the end of
\autoref{sec:undec}).  All type variables $V$ occurring in a valid
signature may refer to each other and must be defined, and all type
variables defined in a signature must be distinct.

\begin{definition}\label{def:unfold}
  We define $\unfold{V} = A$ if $V = A \in \Sigma$ and
  $\unfold{A} = A$ otherwise.
\end{definition}

\begin{definition}\label{def:rel}
  A relation $\rel$ on types is a \emph{type bisimulation} if
  $(A, B) \in \rel$ implies that for $S = \unfold{A}$,
  $T = \unfold{B}$ we have
  \begin{itemize}
  \item If $S = \ichoice{\ell : A_\ell}_{\ell \in L}$ then
    $T = \ichoice{\ell : B_\ell}_{\ell \in L}$ and
    $(A_\ell, B_\ell) \in \rel$ for all $\ell \in L$.
  \item If $S = \echoice{\ell : A_\ell}_{\ell \in L}$ then
    $T = \echoice{\ell : B_\ell}_{\ell \in L}$ and
    $(A_\ell, B_\ell) \in \rel$ for all $\ell \in L$.
  \item If $S = A_1 \tensor A_2$, then $T = B_1 \tensor B_2$ and $(A_1, B_1) \in \rel$ and 
    $(A_2, B_2) \in \rel$. 
  \item If $S = A_1 \lolli A_2$, then $T = B_1 \lolli B_2$ and $(A_1, B_1) \in \rel$ and 
    $(A_2, B_2) \in \rel$. 
  \item If $S = \one$ then $T = \one$.
  \end{itemize}
\end{definition}

\begin{definition}\label{def:tpeq}
  We say that $A$ is \emph{equal} to $B$, written $A \equiv B$, if there is a
  type bisimulation $\rel$ such that $(A, B) \in \rel$.
\end{definition}

As two simple running examples we use an interface to a queue and the
representation of binary numbers as sequences of bits.
\begin{example}[Queues, v1]
  \label{ex:queues-v1}
  A queue provider offers a choice (indicated by $\with$) of either receiving an
  $\mb{ins}$ label followed by a channel of type $A$ (denoted by $\lolli$)
  to insert into
  the queue, or a $\mb{del}$ label to delete an element from the
  queue.  In the latter case, the queue provider has a choice (indicated by $\oplus$)
  of either responding with
  the label $\mb{none}$ (if there is no element in the queue) and
  closes the channel (indicated by $\one$), or the label $\m{some}$ followed by an element
  of type $A$ (denoted by $\tensor$) and recurses to await the next round of interactions.
  \begin{sill}
    $\queue{A} = \echoice{$\=$\mb{ins} : A \lolli \queue{A},$\\
      \>$\mb{del} : \ichoice{$\=$\mb{none} : \one,$\\
        \>\>$\mb{some} : A \tensor \queue{A}}}$
  \end{sill}
  We view $\queue{A}$ as a family of types, one for each $A$, to avoid
  introducing explicit polymorphic type constructors.
\end{example}

\begin{example}[Binary Numbers, v1]
  \label{ex:bin-v1}
  A process representing a binary number either sends a label $\mb{e}$
  representing the number 0 and closes the channel, or one of the labels $\mb{b0}$ (bit 0) or
  $\mb{b1}$ (bit 1) followed by remaining bits (by recursing). We assume a ``little
  endian'' form, that is, the least significant bit is sent first.
  \begin{sill}
    $\m{bin} = \ichoice{\mb{b0} : \m{bin}, \mb{b1} : \m{bin}, \mb{e} : \one}$
  \end{sill}
  As examples of message sequences along a fixed channel, we would have
  \[
    \begin{array}{ll}
      \mb{e} \semi \m{close} & \mbox{representing $0$} \\
      \mb{b0} \semi \mb{e} \semi \m{close} & \mbox{also representing $0$} \\
      \mb{b0} \semi \mb{b1} \semi \mb{e} \semi \m{close} & \mbox{representing $2$} \\
      \mb{b1} \semi \mb{b0} \semi \mb{b1} \semi \mb{b1} \semi \mb{e} \semi \m{close}
                             & \mbox{representing $13$}                                                      

    \end{array}
  \]
\end{example}

\section{Arithmetic Refinements}
\label{sec:arith}

Before we extend our language of types formally, we revisit the
examples in order to motivate the specific constructs available.  We
write $V[\overline{e}]$ for a type indexed by a sequence of arithmetic
expressions $e$.  Since it has been appropriate for most of our
examples, we restrict ourselves to natural numbers rather than
arbitrary integers.

\begin{example}[Queues, v2]
  The provider of a queue should be constrained to answer $\mb{none}$
  exactly if the queue contains no elements and $\mb{some}$ if it is
  nonempty.  The queue type from \autoref{ex:queues-v1} does not
  express this.  This means a client may need to have some redundant
  branches to account for responses that should be impossible.  We
  now define the type $\queue{A}[n]$ to stand for a queue with $n$
  elements.
  \begin{sill}
    $\queue{A}[n] = \echoice{$\=$\mb{ins} : A \lolli \queue{A}[n+1],$\\
      \>$\mb{del} : \ichoice{$\=$\mb{none} : \tassert{n=0} \one,$\\
        \>\>$\mb{some} : \tassert{n > 0} A \tensor \queue{A}[n-1]}}$
  \end{sill}
  The first branch is easy to understand: if we add an element to a
  queue of length $n$, it subsequently contains $n+1$ elements.  In the
  second branch we \emph{constrain} the arithmetic variable $n$ to be
  equal to $0$ if the provider sends $\mb{none}$ and positive if the
  provider sends $\mb{some}$.  In the latter case, we subtract one from
  the length after an element has been dequeued.
\end{example}

Conceptually, the type $\tassert{\phi}A$ means that the provider must
send a proof of $\phi$, so it corresponds to $\exists p : \phi.\, A$.
A characteristic of \emph{type refinement}, in contrast to fully
dependent types, is that the computation of $A$ can only depend on the
\emph{existence} of a proof $p$, but not on its form.  Since our index
domain is also decidable no actual proof needs to be sent (since one
can be constructed from $\phi$ automatically, if needed), just a
token \emph{asserting} its existence.
There is also a dual constructor $\tassume{\phi} A$ that licenses the
\emph{assumption} of $\phi$, which, conceptually, corresponds to
\emph{receiving} a proof of $\phi$.

\begin{example}[Binary Numbers, v2]
  The indexed type $\m{bin}[n]$ should represent a binary number with
  value $n$.  Because the least significant bit comes first, we
  expect, for example, that
  $\m{bin}[n] = \ichoice{\mb{b0} : \tassert{2\mathbin{|}n} \m{bin}[n/2],
    \ldots}$.  However, while divisibility is available in Presburger
  arithmetic, division itself is not; instead, we can express the
  constraint and the index of the recursive occurrence using
  quantification.
  \begin{sill}
    $\m{bin}[n] = \ichoice{$\=$\mb{b0} : \texists{k} \tassert{n = 2*k} \m{bin}[k],$ \\
      \>$\mb{b1} : \texists{k} \tassert{n = 2*k+1} \m{bin}[k],$ \\
      \>$\mb{e} : \tassert{n = 0} \one}$
  \end{sill}
  As a further refinement, we could rule out leading zeros by adding
  the constraint $n > 0$ in the branch for $\mb{b0}$.
\end{example}

The type $\texists{n}A$ means that the provider must send a natural
number $i$ and proceed at type $A[i/n]$, corresponding to
existential quantification in arithmetic.  The dual
universal quantifier $\tforall{n}A$ requires the provider to receive
a number $i$ and proceed at type $A[i/n]$.



We now extend our definitions to account for these new constructs.
Below, $i$ represents a constant, while $n$ represents a natural number variable.
\[
  \begin{array}{llclll}
    \mbox{Types} & A & ::= & \ldots & & \\
                 & & \mid & \tassert{\phi}A & \mbox{assert $\phi$} & \mbox{continue at type $A$} \\
                 & & \mid & \tassume{\phi}A & \mbox{assume $\phi$} & \mbox{continue at type $A$} \\
                 & & \mid & \texists{n}A & \mbox{send number $i$} & \mbox{continue at type $A[i/n]$} \\
                 & & \mid & \tforall{n}A & \mbox{receive number $i$} & \mbox{continue at type $A[i/n]$} \\
                 & & \mid & V[\overline{e}] & \mbox{variable instantiation} \\[1ex]
    \mbox{Arith. Expressions} & e & ::= & \multicolumn{3}{l}{i \mid e + e \mid e - e \mid i \times e \mid (i\mathbin{|}e) \mid n} \\[1ex]
    \mbox{Arith. Propositions} & \phi & ::= & \multicolumn{3}{l}{e = e \mid e > e \mid \top \mid \bot \mid \phi \land \phi \mid \phi \lor \phi \mid \lnot \phi \mid \texists{n}\phi \mid \tforall{n} \phi} \\[1ex]
    \mbox{Signature} & \Sigma & ::= & \multicolumn{3}{l}{\cdot \mid \Sigma, V[\overline{n}\mid \phi] = A}
  \end{array}
\]
An indexed type definition $V[\overline{n}\mid \phi] = A$ requires
every instance $\overline{e}$ of the sequence of variables
$\overline{n}$ to satisfy $\phi[\overline{e}/\overline{n}]$.  This is
verified statically when a type signature is checked for validity, as
defined below.  We use $\vars$ for a collection of arithmetic
variables and $\cons$ (to signify \emph{constraints}) for an
arithmetic proposition occurring among the antecedents of a judgment.
We then have the following rules defining the validity of signatures
($\vdash \Sigma\; \mi{signature}$), declarations
($\vdash_\Sigma \Sigma' \; \mi{valid}$), and types
($\vars \semi \cons \vdash_\Sigma A\; \mi{valid}$) where $\vars$ is a
collection of arithmetic variables including all free variables in
constraint $\cons$ and type $A$.  We silently rename variables so that
$n$ does not already occur in $\vars$ in the ${\exists}V$ and
${\forall}V$ rules.  We also call upon the semantic entailment
judgment $\vars \semi \cons \proves \phi$ which means that
$\forall \vars.\, \cons \supset \phi$ holds in arithmetic and
$\proves \phi$ abbreviates $\cdot \semi \top \proves \phi$.
\[
  \begin{array}{c}
    \infer[]
    {\vdash \Sigma\; \mi{signature}}
    {\vdash_\Sigma \Sigma\; \mi{valid}}
    \hspace{3em}
    \infer[]
    {\vdash_\Sigma (\cdot)\; \mi{valid}}
    {\mathstrut}
    \hspace{3em}
    \infer[]
    {\vdash_\Sigma \Sigma', V[\overline{n}\mid \phi] = A\; \mi{valid}}
    {\vdash_\Sigma \Sigma'\; \mi{valid}
    & \overline{n} \semi \phi \vdash_\Sigma A\; \mi{valid}
    & A \not= V'[\overline{e}']}
    \\[1em]
    \infer[{?}V]
    {\vars \semi \cons \vdash_\Sigma \tassert{\phi} A\; \mi{valid}}
    {\vars \semi \cons \land \phi \vdash_\Sigma A\; \mi{valid}}
    \hspace{3em}
    \infer[{!}V]
    {\vars \semi \cons \vdash_\Sigma \tassume{\phi} A\; \mi{valid}}
    {\vars \semi \cons \land \phi \vdash_\Sigma A\; \mi{valid}}
    \\[1em]
    \infer[{\exists}V^n]
    {\vars \semi \cons \vdash_\Sigma \texists{n} A\; \mi{valid}}
    {\vars,n \semi \cons \vdash_\Sigma A\; \mi{valid}}
    \hspace{3em}
    \infer[{\forall}V^n]
    {\vars \semi \cons \vdash_\Sigma \tforall{n} A\; \mi{valid}}
    {\vars,n \semi \cons \vdash_\Sigma A\; \mi{valid}}
    \\[1em]
    \infer[\m{tdef}]
    {\vars \semi \cons \vdash_\Sigma V[\overline{e}]\; \mi{valid}}
    {V[\overline{n}\mid \phi] = A \in \Sigma
    & \vars \semi \cons \proves \phi[\overline{e}/\overline{n}]}
  \end{array}
\]
We elide the compositional rules for all the other type constructors.
Since we like to work over natural numbers rather than integers, it is
convenient to assume that every definition $V[\overline{n}] = A$
abbreviates $V[\overline{n} \mid \overline{n} \geq 0] = A$.  This
means that in valid signatures every occurrence $V[\overline{e}]$ is
such that $\overline{e} \geq 0$ follows from the known constraints.

\begin{example}
  The declaration
  \begin{sill}
    $\queue{A}[n] = \echoice{$\=$\mb{ins} : A \lolli \queue{A}[n+1],$\\
      \>$\mb{del} : \ichoice{$\=$\mb{none} : \tassert{n=0} \one,$\\
        \>\>$\mb{some} : \tassert{n > 0} A \tensor \queue{A}[n-1]}}$
  \end{sill}
  is valid because $n \geq 0 \proves n+1 \geq 0$ and
  $n \geq 0 \land n > 0 \proves n-1 \geq 0$.
\end{example}

Unfolding a definition must now substitute for the arithmetic variables
we abstract over.
\begin{definition}\label{def:unfold-arith}
  $\unfold{V[\overline{e}]} = A[\overline{e}/\overline{n}]$ if
  $V[\overline{n}\mid\phi] = A \in \Sigma$ and $\unfold{A} = A$
  otherwise.
\end{definition}

We say that a type is \emph{closed} if it contains no free arithmetic
variables $n$.

\begin{definition}\label{def:rel-arith}
  A relation $\rel$ on closed valid types is a \emph{type
    bisimulation} if $(A, B) \in \rel$ implies that for $S = \unfold{A}$,
  $T = \unfold{B}$ we have the following conditions (in addition
  to those of \autoref{def:rel}):
  \begin{itemize}
  \item If $S = \tassert{\phi}A'$ then $T = \tassert{\psi}B'$ and
    either (i) $\proves \phi$, $\proves \psi$, and $(A',B') \in \rel$, \newline
    or (ii) $\proves \lnot \phi$ and $\proves \lnot \psi$.
  \item If $S = \tassume{\phi}A'$ then $T = \tassume{\psi}B'$ and either
    (i) $\proves \phi$, $\proves \psi$, and $(A',B') \in \rel$, \newline
    or (ii) $\proves \lnot \phi$ and $\proves \lnot \psi$
  \item If $S = \texists{m}A'$ then $T = \texists{n}B'$ and for all
    $i \in \mathbb{N}$, $(A'[i/m], B'[i/n]) \in \rel$.
  \item If $S = \tforall{m}A'$ then $T = \tforall{n}B'$ and for all
    $i \in \mathbb{N}$, $(A'[i/m], B'[i/n]) \in \rel$.
  \end{itemize}
  We also extend the notation $A \equiv B$ to this richer set of
  types.
\end{definition}

An interesting point here is provided by the cases \textit{(ii)} in
the first two clauses.  Because the type must be closed, we know that
$\phi$ and $\psi$ will be either true or false.  If both are false, no
messages can be sent along a channel of either type and therefore the
continuation types $A'$ and $B'$ are irrelevant when considering type equality.

Fundamentally, due to the presence of arbitrary recursion and therefore
non-termination, we always
view a type as a restriction of what a process \emph{might} send or
receive along some channel, but it is neither \emph{required}
to send a message nor \emph{guaranteed} to receive one.  This is
similar to functional programming with unrestricted recursion where an
expression may not return a value.  The definition based on
observability of messages is then somewhat strict, as exemplified by
the next example.

\begin{example}
  Consider
  \begin{sill}
    $\m{bin}[n] = \ichoice{$\=$\mb{b0} : \texists{k} \tassert{n = 2*k} \m{bin}[k],$ \\
      \>$\mb{b1} : \texists{k} \tassert{n = 2*k+1} \m{bin}[k],$ \\
      \>$\mb{e} : \tassert{n = 0} \one}$ \\[1ex]
    $\m{zero} = \ichoice{$\=$\mb{b0} : \texists{k} \tassert{k = 0} \m{zero},$ \\
      \>$\mb{e} : \tassert{0 = 0} \one}$
  \end{sill}
  We might expect $\m{bin}[0] \equiv \m{zero}$, but this is not so.  A
  process of type $\m{bin}[0]$ could send the label $\mb{b1}$
  and maybe even, say, $0$ for $k$ and then just loop forever
  (because there is no proof of $0 = 1$).  The type $\m{zero}$
  can not exhibit this behavior so the types are not equivalent.
\end{example}

In our implementation, missing branches for a choice in process
definitions are reconstructed with a continuation that marks it as
impossible, which is then verified by the type checker.  We found this
simple technique significantly limited the need for subtyping or
explicit definition of types such as $\m{zero}$---instead, we just
work with $\m{bin}[0]$.

 The following properties
of type equality are straightforward.

\begin{lemma}[Properties of Type Equality]
  The relation $\equiv$ is reflexive, symmetric, transitive and a
  congruence on closed valid types.
\end{lemma}

\section{Undecidability of Type Equality}
\label{sec:undec}

We prove the undecidability of type equality by exhibiting a reduction
from an undecidable problem about two counter machines.

The type system allows us to simulate two counter
machines~\cite{TwoCounterMachines}. Intuitively, arithmetic constraints allow us to
model branching zero-tests available in the machine.  This, coupled
with recursion in the language of types, establishes
undecidability. Remarkably, a small fragment of our language
containing only type definitions, internal choice ($\oplus$) and
assertions ($\tassert{\phi} A$) where $\phi$ just contains constraints
$n = 0$ and $n > 0$ is sufficient to prove undecidability.  Moreover,
the proof still applies if we treat types isorecursively. In the
remainder of this section we provide some details of the
undecidability proof.

\begin{definition}[Two Counter Machine]
A two counter machine $\tcm$ is given a sequence
of instructions $\iota_1, \iota_2, \ldots, \iota_m$ where each
instruction is one of the following.

\begin{itemize}
\item ``$\inc{c_j} ; \goto \; k$'' (increment counter $j$ by 1 and
go to instruction $k$)
\item ``$\zeroc{c_j} \; \goto \; k : \dec{c_j} ; \goto \; l$'' (if the value
of the counter $j$ is 0, go to instruction $k$, else decrement the
counter by 1 and go to instruction $l$)
\item ``$\halt$'' (stop computation)
\end{itemize}

A configuration $C$ of the machine $\tcm$ is defined as a
triple $(i, c_1, c_2)$, where $i$ denotes the number of the
current instruction and $c_j$'s denote the value of the counters.
A configuration $C'$ is defined as the successor configuration
of $C$, written as $C \mapsto C'$ if $C'$ is the result of
executing the $i$-th instruction on $C$. If $\ins_i = \halt$, then
$C = (i, c_1, c_2)$ has no successor configuration.
The computation of $\tcm$ is the unique maximal sequence
$\rho = \rho(0) \rho(1) \ldots$ such that $\rho(i) \mapsto
\rho(i+1)$ and $\rho(0) = (1, 0, 0)$. Either $\rho$ is infinite,
or ends in $(i, c_1, c_2)$ such that $\ins_i = \halt$ and $c_1, c_2 \in \mathbb{N}$.
\end{definition}

The \emph{halting problem} refers to determining whether the
computation of a two counter machine $\tcm$ with given initial values
$c_1,c_2 \in \mathbb{N}$ is finite.  Both the halting problem and its
dual, the \emph{non-halting problem}, are undecidable.

\begin{theorem}\label{thm:undec}
  Given a valid
  signature $\Sigma$ and two types $A$ and $B$
  such that $m,n \semi \top \vdash_\Sigma A,B\; \mi{valid}$.  Then it
  is undecidable whether for concrete $i,j \in \mathbb{N}$ we have
  $A[i/m,j/n] \equiv B[i/m,j/n]$.
\end{theorem}

\begin{proof}
  Given a two counter machine, we construct a signature $\Sigma$ and
  two types $A$ and $B$ with free arithmetic variables $m$ and $n$
  such that the computation of the machine starting with initial
  counter values $i$ and $j$ is infinite iff
  $A[i/m,j/n] \equiv B[i/m,j/n]$ in $\Sigma$.

  We define types $T_{\inf} = \ichoice{\ell : T_{\inf}}$ and
  $T_{\inf}' = \ichoice{\ell' : T_{\inf}'}$ for \emph{distinct} labels $\ell$
  and $\ell'$.  Next, for every instruction $\ins_i$, we define
  types $T_i$ and $T_i'$ based on the form of the
  instruction.
  \begin{itemize}
  \item Case ($\ins_i = \inc{c_1} ; \goto \; k$): We define
    \[
      \begin{array}{lcl}
        T_i [c_1, c_2] & = &  \ichoice{\m{inc}_1 : T_k [c_1 + 1, c_2]} \\
        T_i' [c_1, c_2] & = & \ichoice{\m{inc}_1 : T_k' [c_1 + 1, c_2]}
      \end{array}
    \]
  \item Case ($\ins_i = \inc{c_2} ; \goto \; k$): We define
    \[
      \begin{array}{lcl}
        T_i [c_1, c_2] & = & \ichoice{\m{inc}_2 : T_k [c_1, c_2 + 1]} \\
        T_i' [c_1, c_2] & = & \ichoice{\m{inc}_2 : T_k' [c_1, c_2 + 1]}
      \end{array}
    \]
  \item Case ($\ins_i = \zeroc{c_1} \; \goto \; k : \dec{c_1} ; \goto \; l$):
    We define
    \[
      \begin{array}{lcl}
        T_i [c_1, c_2] & = & \ichoice{\m{zero}_1 : \tassert{c_1 = 0} T_k[c_1, c_2],
                             \m{dec}_1 : \tassert{c_1 > 0} T_l [c_1 - 1, c_2]} \\
        T_i' [c_1, c_2] & = & \ichoice{\m{zero}_1 : \tassert{c_1 = 0} T_k' [c_1, c_2],
                              \m{dec}_1 : \tassert{c_1 > 0} T_l' [c_1 - 1, c_2]}
      \end{array}
    \]
  \item Case ($\ins_i = \zeroc{c_2} \; \goto \; k : \dec{c_2} ; \goto \; l$):
    We define
    \[
      \begin{array}{lcl}
        T_i [c_1, c_2] & = & \ichoice{\m{zero}_2 : \tassert{c_2 = 0} T_k [c_1, c_2],
        \m{dec}_2 : \tassert{c_2 > 0} T_l [c_1, c_2 - 1]} \\
        T_i' [c_1, c_2] & = & \ichoice{\m{zero}_2 : \tassert{c_2 = 0} T_k' [c_1, c_2],
        \m{dec}_2 : \tassert{c_2 > 0} T_l' [c_1, c_2 - 1]}
      \end{array}
    \]
  \item Case ($\ins_i = \halt$): We define
    \[
      \begin{array}{lcl}
        T_i [c_1, c_2] & = & T_{\inf} \\
        T_i' [c_1, c_2] & = & T_{\inf}'
      \end{array}
    \]
  \end{itemize}
  We set type $A = T_1[m,n]$ and $B = T_1'[m,n]$.  Now suppose, the
  counter machine $\tcm$ is initialized in the state $(1, i, j)$.  The
  type equality question we ask is whether
  $T_1 [i, j] \equiv T_1' [i, j]$.  The two types only differ at the
  halting instruction. If $\tcm$ does not halt, the two types capture
  exactly the same communication behavior, since the halting
  instruction is never reached and they agree on all other
  instructions.  If $\tcm$ halts, the first type proceeds with label
  $\ell$ and the second with $\ell'$ and they are therefore not equal.
  Hence, the two types are equal iff $\tcm$ does not halt.
\end{proof}

We can easily modify this reduction for an isorecursive interpretation
of types, by wrapping
$\ichoice{\mb{unfold} : \underline{\makebox{\quad}}}$ around 
the right-hand side of each type definition to simulate the $\m{unfold}$ message.  We also
see that a host of other problems are undecidable, such as determining
whether two types with free arithmetic variables are equal for all
instances.  This is the problem that arises while type-checking
parametric process definitions.

\section{A Practical Algorithm for Type Equality}
\label{sec:algorithm}

Despite its undecidability, we have designed a coinductive algorithm
for soundly approximating type equality. Similar to Gay and Hole's
algorithm, it proceeds by attempting to construct a bisimulation.  Due
to the undecidability of the problem, our algorithm can terminate in
three different states: (1) we have succeeded in constructing a
bisimulation, (2) we have found a counterexample to type equality by
finding a place where the types may exhibit different behavior, or (3)
we have terminated the search without a definitive answer.  From the
point of view of type-checking, both (2) and (3) are interpreted as a
failure to type-check (but there is a recourse; see
\autoref{sec:eqtypes}).  Our algorithm is expressed as a set of
inference rules where the execution of the algorithm corresponds to
the bottom-up construction of a deduction.  The algorithm is
deterministic (no backtracking) and the implementation is quite
efficient in practice (see \autoref{sec:impl}).

One of the basic operations in Gay and Hole's algorithm is \emph{loop
  detection}, that is, we have to determine that we have already added
an equation $A \equiv B$ to the bisimulation we are constructing.
Since we must treat \emph{open types}, that is, types with free
arithmetic variables subject to some constraints, determining if we
have considered an equation already becomes a difficult operation.  To
that purpose we make an initial pass over the given type and introduce
fresh \emph{internal names} abstracted over their free type variables
and known constraints.  In the resulting signature defined type
variables and type constructor alternates and we can perform loop
detection entirely on type definitions (whether internal or external).


\begin{example}[Queues, v3]
  After creating internal names $\%i$ for the type of queue we obtain
  the following signature (here parametric in $A$).
  \begin{sill}
    \quad $\queue{A}[n] = {\with}\{\mb{ins} : \%0[n], \mb{del} : \%1[n]\}$ \= \\
    \quad $\%0[n] = A \lolli \queue{A}[n+1]$
    \> \qquad $\%3 = \one$ \\
    \quad $\%1[n] = {\oplus}\{\mb{none} : \%2[n], \mb{some} : \%4[n]\}$
    \> \qquad $\%4[n] = {?}\{n > 0\}.\, \%5[n]$ \\
    \quad $\%2[n] = {?}\{n = 0\}.\, \%3$
    \> \qquad $\%5[n \mid n > 0] = A \tensor \queue{A}[n-1]$
  \end{sill}
\end{example}

Based on the invariants established by internal names, the algorithm
only needs to compare two type variables or two structural types.  The
rules are shown in \autoref{fig:tpeq_rules}.  The judgment has the form
$\vars \semi \cons \semi \G \vdash A \equiv B$ where $\vars$ contains
the free arithmetic variables in the constraints $\cons$ and the types
$A$ and $B$, and $\G$ is a collection of \emph{closures}
$\clo{\vars'}{\cons'}{V_1'[\overline{e_1}'] \equiv
  V_2'[\overline{e_2}']}$.  If a derivation can be constructed,
all ground instances of all closures are included in the
resulting bisimulation (see the proof of \autoref{thm:tpeq_sound}).  A
ground instance
$V_1'[\overline{e_1}'[\sigma']] \equiv V_2'[\overline{e_2}'[\sigma']]$
is given by a substitution $\sigma'$ over variables in $\vars'$ such
that $\proves \cons'[\sigma']$.

\begin{figure}[t]
\begin{mathpar}
  \infer[\oplus]
  {\vars \semi \cons \semi \G \vdash
  \ichoice{\ell : A_\ell}_{\ell \in L} \equiv \ichoice{\ell : B_\ell}_{\ell \in L}}
  {\vars \semi \cons \semi \G \vdash
  A_\ell \equiv B_\ell \quad (\forall \ell \in L)}
  \and
  \infer[\with]
  {\vars \semi \cons \semi \G \vdash
  \echoice{\ell : A_\ell}_{\ell \in L} \equiv \echoice{\ell : B_\ell}_{\ell \in L}}
  {\vars \semi \cons \semi \G \vdash
  A_\ell \equiv B_\ell \quad (\forall \ell \in L)}
  \and
  \infer[\tensor]
  {\vars \semi \cons \semi \G \vdash
  A_1 \tensor A_2 \equiv B_1 \tensor B_2}
  {\vars \semi \cons \semi \G \vdash A_1 \equiv B_1 \qquad
  \vars \semi \cons \semi \G \vdash A_2 \equiv B_2}
  \and
  \infer[\lolli]
  {\vars \semi \cons \semi \G \vdash
  A_1 \lolli A_2 \equiv B_1 \lolli B_2}
  {\vars \semi \cons \semi \G \vdash A_1 \equiv B_1 \qquad
  \vars \semi \cons \semi \G \vdash A_2 \equiv B_2}
  \and
  \infer[\one]
  {\vars \semi \cons \semi \G \vdash \one \equiv \one}
  {}
  \and
  \infer[\tassertop]
  {\vars \semi \cons \semi \G \vdash
  \tassert{\phi}{A} \equiv \tassert{\psi}{B}}
  {\vars \semi \cons \proves \phi \leftrightarrow \psi \qquad
  \vars \semi \cons \wedge \phi \semi \G \vdash A \equiv B}
  \and
  \infer[\tassumeop]
  {\vars \semi \cons \semi \G \vdash
  \tassume{\phi}{A} \equiv \tassume{\psi}{B}}
  {\vars \semi \cons \proves \phi \leftrightarrow \psi \qquad
  \vars \semi \cons \wedge \phi \semi \G \vdash A \equiv B}
  \and
  \infer[\exists^k]
  {\vars \semi \cons \semi \G \vdash
  \texists{m}{A} \equiv \texists{n}{B}}
  {\vars, k \semi \cons \semi \G \vdash
  A[k/m] \equiv B[k/n]}
  \and
  \infer[\forall^k]
  {\vars \semi \cons \semi \G \vdash
  \tforall{m}{A} \equiv \tforall{n}{B}}
  {\vars, k \semi \cons \semi \G \vdash
  A[k/m] \equiv B[k/n]}
  \and
  \infer[\bot]
  {\vars \semi \cons \semi \G \vdash
  A \equiv B}
  {\vars \semi \cons \proves \bot}
  \and
  \infer[\m{refl}]
  {\vars \semi \cons \semi \G \vdash V\indv{e} \equiv V[\overline{e}']}
  {\vars \semi \cons \proves e_1 = e_1' \land \ldots \land e_n = e_n'}
  \and 
  \inferrule*[right = $\m{expd}$]
  {V_1[\overline{v_1} \mid \phi_1] = A \in \Sg \and
  V_2[\overline{v_2} \mid \phi_2] = B \in \Sg \\\\
  \gamma = \clo{\vars}{\cons}{V_1\indv{e_1} \equiv V_2\indv{e_2}}\\\\
  \vars \semi \cons \semi \G, \gamma
  \vdash A[\overline{e_1} / \overline{v_1}] \equiv
  B[\overline{e_2}/\overline{v_2}]}
  {\vars \semi \cons \semi \G \vdash
  V_1 \indv{e_1} \equiv V_2 \indv{e_2}}
  \and
  \inferrule*[right=$\m{def}$]
  {
    \clo{\vars'}{\cons'}{V_1[\overline{e_1}'] \equiv V_2[\overline{e_2}']} \in \G \\
    \vars \semi \cons \proves \exists \vars'. \; \cons' \land \overline{e_1}'
    = \overline{e_1} \land \overline{e_2}' = \overline{e_2}
  }
  {
    \vars \semi \cons \semi \G \vdash V_1 \indv{e_1} \equiv V_2 \indv{e_2}
  }
\end{mathpar}
\caption{Algorithmic Rules for Type Equality}
\label{fig:tpeq_rules}
\end{figure}

The rules for type constructors simply compare the components.  If the
type constructors (or the label sets in the $\oplus$ and $\with$ rules)
do not match, then type equality fails (having constructed a
counterexample to bisimulation) unless the $\bot$ rule applies.  This
rules handles the case where the constraints are contradictory and no
communication is possible.

The rule of reflexivity is needed explicitly here (but not in the
version of Gay and Hole) because due to the incompleteness of the
algorithm we may otherwise fail to recognize type variables with equal
index expressions as equal.

Now we come to the key rules, $\m{expd}$ and $\m{def}$.  In the
$\m{expd}$ rule we expand the definitions of $V_1\indv{e_1}$ and
$V_2\indv{e_2}$, and we also add the closure
$\clo{\vars}{\cons}{V_1\indv{e_1} \equiv V_2\indv{e_2}}$ to $\Gamma$.
Since the equality of $V_1\indv{e_1}$ and
$V_2\indv{e_2}$ must hold for all its ground instances, the
extension of $\Gamma$ with the corresponding closure remembers exactly that.

In the $\m{def}$ rule we close off the derivation successfully if all
instances of the equation $V_1\indv{e_1} \equiv V_2\indv{e_2}$ are
already instances of a closure in $\Gamma$.  This is checked by the
entailment in the second premise,
\(
  \vars \semi \cons \proves \exists \vars'.\, \cons' \land
  \overline{E_1} = \overline{e_1} \land \overline{E_2} = \overline{e_2}
\).
This entailment is verified as a closed $\forall\exists$ arithmetic
formula, even if the original constraints $\cons$ and $\cons'$ do not
contain any quantifiers.  While for Presburger arithmetic we can
decide such a proposition using quantifier elimination, other
constraint domains may not permit such a decision procedure.

The algorithm so far is sound, but potentially nonterminating because
when encountering variable/variable equations, we can use the
$\m{expd}$ rule indefinitely.  To ensure termination, we restrict the
$\m{expd}$ rule to the case where \emph{no} formula with the same type
variables $V_1$ and $V_2$ is already present in $\Gamma$.  This also
removes the overlap between these two rules.  Note that if type
variables have no parameters, our algorithm specializes to Gay and
Hole's (with the small optimizations of reflexivity and internal
naming), which means our algorithm is sound and complete on unindexed
types.


\begin{example}[Integer Counter]
  An integer counter with increment ($\mb{inc}$), decrement ($\mb{dec}$)
  and sign-test ($\mb{sgn}$) operations provides type $\m{intctr}[x,y]$,
  where the current value of the counter is $x-y$ for natural numbers $x$
  and $y$.
  \begin{sill}
    \quad $\m{intctr}[x,y] = {\with}\{$\=$\mb{inc} : \m{intctr}[x+1,y],$ \\
    \>$\mb{dec} : \m{intctr}[x,y+1],$ \\
    \>$\mb{sgn} : {\oplus}\{$\=$\mb{neg} : \tassert{x < y} \m{intctr}[x,y],$ \\
    \>\>$\mb{zer} : \tassert{x = y} \m{intctr}[x,y],$ \\
    \>\>$\mb{pos} : \tassert{x > y} \m{intctr}[x,y]\}\}$
  \end{sill}
  Under this definition our algorithm verifies, for example, that
  an increment followed by a decrement does not change the counter value.
  That is,
  \[
    \begin{array}{l}
      x,y \semi \top \semi \cdot \vdash \m{intctr}[x,y] \equiv \m{intctr}[x+1,y+1]
    \end{array}
  \]
  where we have elided the assumptions $x,y \geq 0$.
  When applying $\m{expd}$, we assume
  $\gamma = \clo{x',y'}{\top}{\m{intctr}[x',y'] \equiv \m{intctr}[x'+1,y'+1]}$.  Then, for
  example, in the first branch (for $\m{inc}$) we conclude
  $x,y \semi \top \semi \gamma \vdash \m{intctr}[x+1,y] \equiv \m{intcr}[x+2,y+1]$ using the $\m{def}$ rule and
  the entailment $x,y \semi \top \proves \exists x'.\, \exists y'.\,
  x' = x+1 \land y' = y \land x'+1 = x+2 \land y'+1 = y+1$. The other branches
  are similar.
  \end{example}

\subsection{Soundness of the Type Equality Algorithm}

We prove that the type equality algorithm is sound with respect to the
definition of type equality. The soundness is proved by constructing a
type bisimulation from a derivation of the algorithmic type equality
judgment.  We sketch the key points of the proofs.

The first gap we have to bridge is that the type bisimulation is
defined only for closed types, because observations can only arise
from communication along channels which, at runtime, will be of closed
type.  So, if we can derive
$\vars \semi \cons \semi \cdot \vdash A \equiv B$ then we should
interpret this as stating that for all ground substitutions $\sigma$
over $\vars$ such that $\proves \cons[\sigma]$ we have
$A[\sigma] \equiv B[\sigma]$.

\begin{definition}\label{def:app_algo_tpeq}
  Given a relation $\rel$ on valid ground types and two types $A$ and
  $B$ such that $\vars \semi \cons \vdash A,B\; \mi{valid}$, we write
  $\forall \vars.\, \cons \Rightarrow A \equiv_\rel B$ if for all
  ground substitutions $\sigma$ over $\vars$ such that
  $\proves \cons[\sigma]$ we have $(A[\sigma], B[\sigma]) \in \rel$.

  Furthermore, we write $\forall \vars.\, \cons \Rightarrow
  A \equiv B$ if there exists a type bisimulation $\rel$ such
  that $\forall \vars.\, \cons \Rightarrow A \equiv_\rel B$. 
\end{definition}

Note that if $\vars \semi \cons \proves \bot$, then
$\forall \vars.\, \cons \Rightarrow A \equiv B$ is vacuously true,
since there does not exist a ground substitution $\sigma$ such that
$\proves \cons[\sigma]$.  A key lemma is the following, which is
needed to show the soundness of the $\m{def}$ rule.

\begin{lemma}\label{lem:app_gen_sim}
  Suppose $\forall \vars'. \cons' \Rightarrow V_1 [\overline{e_1}'] \equiv_\rel V_2[\overline{e_2}']$ holds.
  Further assume that $\vars \semi \cons \proves \exists \vars'. \cons' \wedge
  \overline{e_1}' = \overline{e_1} \wedge \overline{e_2}' = \overline{e_2}$
  for some $\vars, \cons, \overline{e_1}, \overline{e_2}$. Then, $\forall \vars. \cons
  \Rightarrow V_1 \indv{e_1} \equiv_\rel V_2 \indv{e_2}$ holds.
\end{lemma}

\begin{proof}
  To prove $\forall \vars.\, \cons \Rightarrow V_1 \indv{e_1} \equiv_\rel V_2 \indv{e_2}$, it
  is sufficient to show that $V_1[\overline{e_1}[\sigma]] \equiv_\rel V_2[\overline{e_2}[\sigma]]$
  for any substitution $\sigma$ over $\vars$ such that $\proves \cons[\sigma]$.
  Applying this substitution to $\vars \semi \cons \proves \exists \vars'.\,
  \cons' \wedge \overline{e_1}' = \overline{e_1} \wedge \overline{e_2}' = \overline{e_2}$,
  we infer $\exists \vars'.\, \cons' \wedge \overline{e_1}' =
  \overline{e_1}[\sigma] \wedge \overline{e_2}' = \overline{e_2}[\sigma]$
  since $\proves \cons[\sigma]$. Thus, there exists $\sigma'$ over $\vars'$ such that
  $\proves \cons'[\sigma']$ holds, and $\overline{e_1}'[\sigma'] =
  \overline{e_1}[\sigma]$ and $\overline{e_2}'[\sigma'] = \overline{e_2}[\sigma]$.
  And since $\forall \vars'.\, \cons' \Rightarrow V_1[\overline{e_1}'] \equiv_\rel V_2[\overline{e_2}']$,
  we deduce that for any ground substitution (including the current one)
  $\sigma'$ over $\vars'$, $V_1[\overline{e_1}'[\sigma']] \equiv_\rel
  V_2[\overline{e_2}'[\sigma']]$ holds. This implies that $V_1[\overline{e_1}[\sigma]] \equiv_\rel
  V_2[\overline{e_2}[\sigma]]$ since $\overline{e_1}'[\sigma'] =
  \overline{e_1}[\sigma]$ and $\overline{e_2}'[\sigma'] = \overline{e_2}[\sigma]$.
\end{proof}

We construct the bisimulation from a derivation of
$\vars \semi \cons \semi \G \vdash A \equiv B$ by (i) collecting the
conclusions of all the sequents, excepting only the $\m{def}$ rule,
and (ii) forming all ground instances from them.

\begin{definition}
  Given a derivation $\mathcal{D}$ of
  $\vars \semi \cons \semi \G \vdash A \equiv B$, we define the set
  $\mathcal{S}(\mathcal{D})$ of closures.  For each sequent
  $\vars' \semi \cons' \semi \G' \vdash A' \equiv B'$ (except the
  conclusion of the $\m{def}$ rule) we include the closure
  $\clo{\vars'}{\cons'}{A' \equiv B'}$ in $\mathcal{S}(\mathcal{D})$.
\end{definition}

\begin{theorem}\label{thm:tpeq_sound}
  If $\vars \semi \cons \semi \cdot \vdash A \equiv B$, then
  $\forall \vars.\, \cons \Rightarrow A \equiv B$.
\end{theorem}

\begin{proof}
  We are given a derivation $\mathcal{D}_0$ of $\vars_0 \semi \cons_0 \semi \cdot \vdash A_0 \equiv B_0$.
  Construct $\mathcal{S}(\mathcal{D}_0)$ and define a relation $\rel$ on closed
  valid types as follows:
  \begin{eqnarray*}
    \rel = \{(A[\sigma], B[\sigma]) \mid \clo{\vars}{\cons}{A \equiv B}
    \in \mathcal{S}(\mathcal{D}_0)\;\mbox{and}\;\sigma\;\mbox{over}\;\vars\;\mbox{with}\;\proves \cons[\sigma]\}
  \end{eqnarray*}
  We prove that $\rel$ is a type bisimulation. Then our theorem follows
  since the closure
  $\clo{\vars_0}{\cons_0}{A_0 \equiv B_0} \in \mathcal{S}(\mathcal{D}_0)$.

  Consider $(A[\sigma], B[\sigma]) \in \rel$ where
  $\clo{\vars}{\cons}{A \equiv B} \in \mathcal{S}(\mathcal{D}_0)$ for
  some $\sigma$ over $\vars$ and $\proves \cons[\sigma]$.
  
  First, consider the case where $\vars \semi \cons \proves \bot$. Under such a constraint,
  $\vars \semi \cons \semi \cdot \vdash A \equiv B$ holds true due to the $\bot$ rule.
  Furthermore, $\forall \vars.\, \cons \Rightarrow A \equiv B$ holds vacuously, and
  the algorithm is sound. For the remaining cases, we case analyze on the structure of $A[\sigma]$
  and assume that there exists a ground substitution $\sigma$ such that $\proves \cons[\sigma]$.

  Consider the case where $A = \ichoice{\ell : A_\ell}_{\ell \in L}$. Since $A$ and $B$
  are both structural, $B = \ichoice{\ell : B_\ell}_{\ell \in L}$. Since
  $\clo{\vars}{\cons}{A \equiv B} \in \mathcal{S}(\mathcal{D}_0)$, by definition of
  $\mathcal{S}(\mathcal{D}_0)$, we get
  $\clo{\vars}{\cons}{A_\ell \equiv B_\ell} \in \mathcal{S}(\mathcal{D}_0)$ for all
  $\ell \in L$. By the definition of $\rel$, we get that
  $(A_\ell[\sigma], B_\ell[\sigma]) \in \rel$. Also,
  $A[\sigma] = \ichoice{\ell : A_\ell[\sigma]}_{\ell \in L}$ and similarly,
  $B[\sigma] = \ichoice{\ell : B_\ell[\sigma]}_{\ell \in L}$.  Hence, $\rel$ satisfies the
  appropriate closure condition for a type bisimulation.

  Next, consider the case where $A = \tassert{\phi} A'$. Since $A$ and $B$ are both
  structural, $B = \tassert{\psi} B'$. Since
  $\clo{\vars}{\cons}{A \equiv B} \in \mathcal{S}(\mathcal{D}_0)$, we obtain
  $\vars \semi \cons \proves \phi \leftrightarrow \psi$ and
  $\clo{\vars}{\cons \land \phi}{A' \equiv B'} \in \mathcal{S}(\mathcal{D}_0)$. Thus, for
  any substitution $\sigma$ such that $\proves \cons[\sigma] \land \phi [\sigma]$,
  we get that $(A'[\sigma], B'[\sigma]) \in \rel$ with
  $A[\sigma] = \tassert{\phi[\sigma]} A'[\sigma]$ and
  $B[\sigma] = \tassert{\psi[\sigma]} B'[\sigma]$. Since $\proves \phi[\sigma]$ and
  and $\vars \semi \cons \proves \phi \leftrightarrow \psi$ we also obtain $\proves \psi[\sigma]$
  and the closure condition is satisfied.

  Next, consider the case where $A = \texists{m} A'$. Since $A$ and $B$ are both
  structural, $B = \texists{n} B'$. Since $\clo{\vars}{\cons}{A \equiv B}
  \in \mathcal{S}(\mathcal{D}_0)$, we get that $\clo{\vars,k}{\cons}{A'[k/m] \equiv B'[k/n]}
  \in \mathcal{S}(\mathcal{D}_0)$. Since $k$ was chosen fresh and does not occur in
  $\cons$, we obtain that for any $i \in \mathbb{N}$ we have $\proves \cons[\sigma, i/k]$
  and therefore $(A'[\sigma, i/k], B'[\sigma, i/k]) \in \rel$ for all $i \in \mathbb{N}$ and
  the closure condition is satisfied.


  The only case where a conclusion is not added to $\mathcal{S}(\mathcal{D}_0)$ is the $\m{def}$
  rule. In this case, adding $(\forall \vars.\, \cons \Rightarrow V_1 \indv{e_1} \equiv V_2
  \indv{e_2})$ is redundant: \autoref{lem:app_gen_sim} states that
  $V_1[\overline{e_1}[\sigma]] \equiv_\rel V_2 [\overline{e_2}[\sigma]]$ which implies
  $(V_1[\overline{e_1}[\sigma]], V_2[\overline{e_2}[\sigma]]) \in \rel$.
\end{proof}

\subsection{Type Equality Declarations}
\label{sec:eqtypes}

Even though the type equality algorithm in \autoref{sec:algorithm} is incomplete, we have
yet to find a natural example where it fails after we added reflexivity as a general rule.
But since we cannot see a simple reason why this should be so, we made our type equality
algorithm extensible by the programmer via an additional form of declaration \[
  \forall \vars.\, \cons \Rightarrow V_1[\overline{e_1}] \equiv V_2[\overline{e_2}]
\]
in signatures.  Let $\Gamma_\Sigma$ denote the set of all such
declarations.  Then we check
\[
  \vars \semi \cons \semi \Gamma_\Sigma \vdash V_1[\overline{e_1}] \equiv V_2[\overline{e_2}]
\]
for each such declaration, seeding the construction of a bisimulation
with all the given equations.  Then, when type-checking has to decide
the equality of two types, it starts not with the empty context
$\Gamma$ but with $\Gamma_\Sigma$.  Our soundness proof can easily
accommodate this more general algorithm.




\section{Implementation and Further Examples}
\label{sec:impl}

We have implemented the algorithm presented in \autoref{sec:algorithm} as part of the Rast
programming language \cite{Das20fscd}, whose name derives from ``Resource-Aware Session
Types''. Rast is based on intuitionistic linear
sessions~\cite{Caires10concur,Caires16mscs} extended with general equirecursive types
and recursively defined processes. We do not explicitly dualize types~\cite{Wadler12icfp}
but distinguish providers and clients that are connected by a private channel. In parallel
work we have proved type safety for Rast, which includes type preservation (session
fidelity) and global progress (deadlock freedom). The open-source implementation is
written in Standard ML and currently comprises about 7500 lines of source
code~\cite{Rast20repo}.

Rast supports indexed types, quantifiers, and arithmetic constraints, following the
presentation in this paper with minor syntactic differences. In addition, Rast has
temporal~\cite{Das18icfp} and ergometric~\cite{Das18lics} types that capture parallel
and sequential complexity of programs. These bounds often depend on intrinisic
properties of the data structures (such as the length of a queue or the value of a binary
number) which are expressed as arithmetic indices.

Rast's linear type checker is bidirectional, which means that only process definitions
need to be annotated with their types. In the so-called \emph{explicit syntax} type
checking is then straightforward, breaking down the structure of the type and unfolding
definitions, except for calls to type equality (which are necessary for forwarding,
process invocations, and sending of channels). The implementation also supports an
\emph{implicit syntax} in which some parts of the program, specifically those asserting or
assuming constraints and a few others relevant to resource analysis, can be omitted from
the source and are reconstructed. The reconstructed code is then passed through the type
checker as ultimate arbiter.

We use a straightforward implementation of Cooper's algorithm~\cite{cooper1972theorem} to
decide Presburger arithmetic with two small but significant optimizations. One takes
advantage of the fact that we are working over natural numbers rather than integers, the
other is to eliminate constraints of the form $x = e$ by substituting $e$ for $x$ in order
to reduce the number of variables. We also extend our solver to handle non-linear constraints.
Since non-linear arithmetic is undecidable, in general, we use a normalizer which collects
coefficients of each term in the multinomial expression.
To check $e_1 = e_2$, we normalize $e_1 - e_2$ and check that
each coefficient of the normal form is $0$.
To check $e_1 \geq e_2$, we normalize $e_1 - e_2$ and check that
each coefficient is non-negative.


\begin{table}[t]
  \centering
  \begin{tabular}{l r r r r r}
  \textbf{Module} & \textbf{iLOC} & \textbf{eLOC} & \textbf{\#Defs} & \textbf{R (ms)} & \textbf{T (ms)} \\
  \toprule
  arithmetic & 69 & 143 & 8 & 0.353 & 1.325 \\
  integers & 90 & 114 & 8 & 0.200 & 1.074 \\
  linlam & 54 & 67 & 6 & 0.734 & 4.003 \\
  list & 244 & 441 & 29 & 1.534 & 3.419 \\
  primes & 90 & 118 & 8 & 0.196 & 1.646 \\
  segments & 48 & 65 & 9 & 0.239 & 0.195 \\
  ternary & 156 & 235 & 16 & 0.550 & 1.967 \\
  theorems & 79 & 141 & 16 & 0.361 & 0.894 \\
  tries & 147 & 308 & 9 & 1.113 & 5.283 \\
  \midrule
  \textbf{Total} & \textbf{977} & \textbf{1632} & \textbf{109} & \textbf{5.280} & \textbf{19.806} \\
  \bottomrule \\
  \end{tabular}
  \caption{Case Studies}
  \label{tab:case_study}
\end{table}

We have a variety of 21 examples implemented, totaling about 3700 lines of code,
for which complete code can be found in our open source repository~\cite{Rast20repo}.
Table~\ref{tab:case_study} describes the results for nine representative case studies: iLOC
describes the lines of source code in implicit syntax, eLOC describes the lines of code
after reconstruction, \#Defs shows the number of process definitions, R (ms) and T (ms)
show the reconstruction and type-checking time in milliseconds respectively. The
experiments were run on an Intel Core i5 2.7 GHz processor with 16 GB 1867 MHz DDR3
memory. We briefly describe each case study.

\begin{enumerate}
\item \textbf{arithmetic}: natural numbers in unary and binary
  representation indexed by their value and processes implementing
  standard arithmetic operations.

\item \textbf{integers}: an integer counter represented using two indices
  $x$ and $y$ with value $x-y$.

\item \textbf{linlam}: expressions in the linear $\lambda$-calculus
  indexed by their size with an \emph{eval} process to evaluate them
  (see below for an excerpt).

\item \textbf{list}: lists indexed by their size
  with standard operations (e.g., \emph{append, reverse, map}).

\item \textbf{primes}: implementation of the sieve of Eratosthenes.

\item \textbf{segments}: type $\m{seg}[n] =
  \forall k. \m{list}[k] \lolli \m{list}[n+k]$ representing partial lists
  with constant-work append operation.

\item \textbf{ternary}: natural numbers represented in balanced
  ternary form with digits $0, 1, -1$, indexed by their value, and
  some standard operations on them.

\item \textbf{theorems}: processes representing (circular~\cite{Derakhshan19arxiv})
  proofs of simple arithmetic theorems.

\item \textbf{tries}: a trie data structure to store multisets of
  binary numbers, with constant amortized work insertion and deletion,
  verified with ergometric types.
\end{enumerate}

\subparagraph*{Linear $\lambda$-calculus}

We briefly sketch the types in an implementation of the (untyped) linear
$\lambda$-calculus in which the index objects track the size of the expression,
because it uses multiple feature of the type system.
\begin{sill}
  $\m{exp}[n] = \ichoice{$\=$\mb{lam} : \tassert{n > 0}
    \tforall{n_1} \m{exp}[n_1] \lolli \m{exp}[n_1+n-1],$\\
    \>$\mb{app} : \texists{n_1} \texists{n_2} \tassert{n = n_1+n_2+1} \m{exp}[n_1] \tensor \m{exp}[n_2]}$
\end{sill}
An expression is either a $\lambda$-abstraction (sending label $\mb{lam}$) or an
application (sending label $\mb{app}$). In case of $\mb{lam}$, the continuation receives a
number $n_1$ and an argument of size $n_1$ and then behaves like the body of the
$\lambda$-abstraction of size $n_1+n-1$. In case of $\mb{app}$, it will send $n_1$ and
$n_2$ such that $n = n_1 +n_2+1$ followed an expression of size $n_1$ and then behave as
an expression of size $n_2$.

A value can only be a $\lambda$-abstraction
\begin{sill}
$\m{val}[n] = \ichoice{$\=$\mb{lam} : \tassert{n > 0}
  \tforall{n_1} \m{exp}[n_1] \lolli \m{exp}[n_1+n-1]}$
\end{sill}
so the $\mb{app}$ label is not permitted. Type checking verifies that that the result of
evaluating a linear $\lambda$-term is no larger than the original term. The declaration
below expresses that $\mi{eval}\,[n]$ is client to a process sending a
$\lambda$-expression of size $n$ along channel $e$ and provides a value of size $k$, where
$k \leq n$.
\begin{sill}
  $(e : \m{exp}[n]) \vdash \mi{eval}\,[n] :: (v : \texists{k} \tassert{k \leq n}
  \m{val}[k])$
\end{sill}

\section{Further Related Work}
\label{sec:related}

Traditional languages with dependent type refinements such as
Zenger's~\cite{Zenger97} or DML~\cite{Xi99popl} only use the rule of
reflexivity as a criterion for equality of indexed types.  This is
justified in the context of these functional languages because data
types are generative and therefore \emph{nominal} in nature.  This is
also true for more recent languages with linearity and value-dependent
types such as Granule~\cite{Orchard19icfp}.

Session type systems that allow dependencies are label-dependent
session types~\cite{Thiemann20popl} and
richer linear type
theories~\cite{Toninho11ppdp,Pfenning11cpp,Toninho18fossacs}.  Toninho
et al.~\cite{Toninho11ppdp,Pfenning11cpp} allow sufficient
dependencies that, in general, proofs must be sent in some
circumstances.  They do not provide a type equality algorithm or
implementation.  In a more recent paper, Toninho et
al.~\cite{Toninho18fossacs} propose a dependent type theory with rich
notions of value and process equality based on $\beta\eta$-congruences
and certain process equalities, but they do not discuss decidability
or algorithms for type checking or type equality.  Wu and
Xi~\cite{Wu17arxiv} propose a dependent session type system
based on ATS~\cite{Xi03types} formalizing type equality in
terms of subtyping and regular constraint relations.  They mention
recursive session types as a possible extension, but do not develop
them nor investigate properties of the required type equality.

Linearly refined session types~\cite{Baltazar12linearity,Franco13sefm}
extend the $\pi$-calculus with capabilities from a fragment of
multiplicative linear logic. These capabilities encode an
authorization logic enabling fine-grained specifications and are thus
not directly comparable to arithmetic
refinements.
Session types with limited arithmetic refinements (only base
types could be refined) have been proposed for the purpose of runtime
monitoring~\cite{Gommerstadt18esop,Gommerstadt19phd}, which is
complementary to our uses for static verification.  They have also
been proposed to capture work~\cite{Das18lics,Das19Nomos} and parallel
time~\cite{Das18icfp}, but parameterization over index objects was
left to an informal meta-level and not part of the object language.
Consequently, these languages contain neither constraints nor
quantifiers, and the metatheory of type equality, type checking, and
reconstruction in the presence of index variables was not developed.

Context-free session types~\cite{Thiemann16icfp} are another generalization of
basic session types in a different direction, essentially allowing the
concatenation of sessions. This generalization has decidable type checking and
type equality problems that have been shown to be efficient in
practice~\cite{Almeida20tacas}.

Asynchronous session types~\cite{Gay10jfp} have a notion of subtyping
under different assumptions regarding communication
behavior~\cite{Mostrous15iandc}.  The resulting subtyping relation
also turns out to be undecidable~\cite{Bravetti17iandc,Lange17fossacs}
with the development of recent practical incomplete
algorithms~\cite{Bravetti19concur}.  The expressive power of
asynchronous session subtyping seems incomparable to our
arithmetically refined session types.

\section{Conclusion}
\label{sec:conc}

This paper explored the metatheory of session types with arithmetic
refinements, showing the undecidability of type equality.
Nevertheless, we have shown a sound, but incomplete algorithm that has
performed well over a range of examples in our Rast implementation.

Natural extensions include nonlinear arithmetic and other constraint
domains, balancing practicality of type checking with expressive
power.  We would also like to generalize from type equality to
subtyping, replacing the notion of bisimulation with a simulation.
Clearly, this will be undecidable as well, but the pioneering work by
Gay and Hole and the characteristics of our algorithms suggest that it
should extend cleanly and remain practical.

Finally, we would also like to generalize our approach to a mixed
linear/nonlinear language~\cite{Benton94csl} or all the way to adjoint
session types~\cite{Pfenning15fossacs,Pruiksma19places}.  Since the
main issues of type equality are orthogonal to the presence or absence
of structural properties, we conjecture that the algorithm proposed
here will extend to this more general setting.

\bibliography{refs}

\end{document}